\documentclass[a4paper,twocolumn,11pt,accepted=2023-08-29]{quantumarticle}
\pdfoutput=1
\usepackage[utf8]{inputenc}
\usepackage[english]{babel}
\usepackage[T1]{fontenc}
\usepackage{amsmath}
\usepackage{hyperref}
\usepackage[numbers,sort&compress]{natbib}
\usepackage{tikz}
\usepackage{lipsum}
\usepackage{amsfonts,amssymb,caption,color,epsfig,graphics,graphicx,hyperref,latexsym,mathrsfs,revsymb,theorem,url,verbatim,enumerate,epstopdf,float,multirow,booktabs,appendix}

\newtheorem{definition}{Definition} 
\newtheorem{proposition}{Proposition} 
\newtheorem{lemma}{Lemma}

\newtheorem{theorem}{Theorem} 
\newtheorem{corollary}[definition]{Corollary}
\newtheorem{conjecture}{Conjecture}

\newtheorem{remark}[definition]{Remark}
\newtheorem{example}{Example} 
\newtheorem{question}[definition]{Question}

\def\bcj{\begin{conjecture}}
	\def\ecj{\end{conjecture}}
\def\bcr{\begin{corollary}}
	\def\ecr{\end{corollary}}
\def\bd{\begin{definition}}
	\def\ed{\end{definition}}
\def\bea{\begin{eqnarray}}
	\def\eea{\end{eqnarray}}
\def\bem{\begin{enumerate}}
	\def\eem{\end{enumerate}}
\def\bex{\begin{example}}
	\def\eex{\end{example}}
\def\bim{\begin{itemize}}
	\def\eim{\end{itemize}}
\def\bl{\begin{lemma}}
	\def\el{\end{lemma}}
\def\bma{\begin{bmatrix}}
	\def\ema{\end{bmatrix}}
\def\bpf{\begin{proof}}
	\def\epf{\end{proof}}
\def\bpp{\begin{proposition}}
	\def\epp{\end{proposition}}
\def\bqu{\begin{question}}
	\def\equ{\end{question}}
\def\br{\begin{remark}}
	\def\er{\end{remark}}
\def\bt{\begin{theorem}}
	\def\et{\end{theorem}}

\def\squareforqed{\hbox{\rlap{$\sqcap$}$\sqcup$}}
\def\qed{\ifmmode\squareforqed\else{\unskip\nobreak\hfil
		\penalty50\hskip1em\null\nobreak\hfil\squareforqed
		\parfillskip=0pt\finalhyphendemerits=0\endgraf}\fi}
\def\endenv{\ifmmode\;\else{\unskip\nobreak\hfil
		\penalty50\hskip1em\null\nobreak\hfil\;
		\parfillskip=0pt\finalhyphendemerits=0\endgraf}\fi}
\newenvironment{proof}{\noindent \textbf{{Proof.~} }}{\qed}
\def\Dbar{\leavevmode\lower.6ex\hbox to 0pt
	{\hskip-.23ex\accent"16\hss}D}
\makeatletter
\def\url@leostyle{%
	\@ifundefined{selectfont}{\def\UrlFont{\sf}}{\def\UrlFont{\small\ttfamily}}}
\makeatother
\def\bcj{\begin{conjecture}}
	\def\ecj{\end{conjecture}}
\def\bcr{\begin{corollary}}
	\def\ecr{\end{corollary}}
\def\bd{\begin{definition}}
	\def\ed{\end{definition}}
\def\bea{\begin{eqnarray}}
	\def\eea{\end{eqnarray}}
\def\bem{\begin{enumerate}}
	\def\eem{\end{enumerate}}
\def\bex{\begin{example}}
	\def\eex{\end{example}}
\def\bim{\begin{itemize}}
	\def\eim{\end{itemize}}
\def\bl{\begin{lemma}}
	\def\el{\end{lemma}}
\def\bpf{\begin{proof}}
	\def\epf{\end{proof}}
\def\bpp{\begin{proposition}}
	\def\epp{\end{proposition}}
\def\bqu{\begin{question}}
	\def\equ{\end{question}}
\def\br{\begin{remark}}
	\def\er{\end{remark}}
\def\bt{\begin{theorem}}
	\def\et{\end{theorem}}

\def\btb{\begin{tabular}}
	\def\etb{\end{tabular}}

\newcommand{\nc}{\newcommand}

\nc{\bbA}{\mathbb{A}} \nc{\bbB}{\mathbb{B}} \nc{\bbC}{\mathbb{C}}
\nc{\bbD}{\mathbb{D}} \nc{\bbE}{\mathbb{E}} \nc{\bbF}{\mathbb{F}}
\nc{\bbG}{\mathbb{G}} \nc{\bbH}{\mathbb{H}} \nc{\bbI}{\mathbb{I}}
\nc{\bbJ}{\mathbb{J}} \nc{\bbK}{\mathbb{K}} \nc{\bbL}{\mathbb{L}}
\nc{\bbM}{\mathbb{M}} \nc{\bbN}{\mathbb{N}} \nc{\bbO}{\mathbb{O}}
\nc{\bbP}{\mathbb{P}} \nc{\bbQ}{\mathbb{Q}} \nc{\bbR}{\mathbb{R}}
\nc{\bbS}{\mathbb{S}} \nc{\bbT}{\mathbb{T}} \nc{\bbU}{\mathbb{U}}
\nc{\bbV}{\mathbb{V}} \nc{\bbW}{\mathbb{W}} \nc{\bbX}{\mathbb{X}}
\nc{\bbZ}{\mathbb{Z}}

\nc{\bA}{{\bf A}} \nc{\bB}{{\bf B}} \nc{\bC}{{\bf C}}
\nc{\bD}{{\bf D}} \nc{\bE}{{\bf E}} \nc{\bF}{{\bf F}}
\nc{\bG}{{\bf G}} \nc{\bH}{{\bf H}} \nc{\bI}{{\bf I}}
\nc{\bJ}{{\bf J}} \nc{\bK}{{\bf K}} \nc{\bL}{{\bf L}}
\nc{\bM}{{\bf M}} \nc{\bN}{{\bf N}} \nc{\bO}{{\bf O}}
\nc{\bP}{{\bf P}} \nc{\bQ}{{\bf Q}} \nc{\bR}{{\bf R}}
\nc{\bS}{{\bf S}} \nc{\bT}{{\bf T}} \nc{\bU}{{\bf U}}
\nc{\bV}{{\bf V}} \nc{\bW}{{\bf W}} \nc{\bX}{{\bf X}}
\nc{\bZ}{{\bf Z}} \nc{\bm}{{\bf m}} \nc{\bv}{{\bf v}}
\nc{\ba}{{\bf a}} \nc{\be}{{\bf e}} \nc{\bu}{{\bf u}}
\nc{\brr}{{\bf r}}

\nc{\cA}{{\cal A}} \nc{\cB}{{\cal B}} \nc{\cC}{{\cal C}}
\nc{\cD}{{\cal D}} \nc{\cE}{{\cal E}} \nc{\cF}{{\cal F}}
\nc{\cG}{{\cal G}} \nc{\cH}{{\cal H}} \nc{\cI}{{\cal I}}
\nc{\cJ}{{\cal J}} \nc{\cK}{{\cal K}} \nc{\cL}{{\cal L}}
\nc{\cM}{{\cal M}} \nc{\cN}{{\cal N}} \nc{\cO}{{\cal O}}
\nc{\cP}{{\cal P}} \nc{\cQ}{{\cal Q}} \nc{\cR}{{\cal R}}
\nc{\cS}{{\cal S}} \nc{\cT}{{\cal T}} \nc{\cU}{{\cal U}}
\nc{\cV}{{\cal V}} \nc{\cW}{{\cal W}} \nc{\cX}{{\cal X}}
\nc{\cZ}{{\cal Z}}

\nc{\hA}{{\hat{A}}} \nc{\hB}{{\hat{B}}} \nc{\hC}{{\hat{C}}}
\nc{\hD}{{\hat{D}}} \nc{\hE}{{\hat{E}}} \nc{\hF}{{\hat{F}}}
\nc{\hG}{{\hat{G}}} \nc{\hH}{{\hat{H}}} \nc{\hI}{{\hat{I}}}
\nc{\hJ}{{\hat{J}}} \nc{\hK}{{\hat{K}}} \nc{\hL}{{\hat{L}}}
\nc{\hM}{{\hat{M}}} \nc{\hN}{{\hat{N}}} \nc{\hO}{{\hat{O}}}
\nc{\hP}{{\hat{P}}} \nc{\hR}{{\hat{R}}} \nc{\hS}{{\hat{S}}}
\nc{\hT}{{\hat{T}}} \nc{\hU}{{\hat{U}}} \nc{\hV}{{\hat{V}}}
\nc{\hW}{{\hat{W}}} \nc{\hX}{{\hat{X}}} \nc{\hZ}{{\hat{Z}}}

\nc{\hn}{{\hat{n}}}

\def\max{\mathop{\rm max}}

\def\tr{\mathop{\rm Tr}}

\newcommand{\ket}[1]{|#1\rangle}

\def\Dbar{\leavevmode\lower.6ex\hbox to 0pt
	{\hskip-.23ex\accent"16\hss}D}

	\newcommand{\add}[1]{{\color{black}#1}}

\begin{document}

\title{Strongly nonlocal unextendible product bases do exist}

	\title{Bounds on the smallest sets of  quantum states with  special quantum nonlocality  }

	\author{Mao-Sheng Li}
	\email{li.maosheng.math@gmail.com}
	
	\affiliation{ School of Mathematics,
		South China University of Technology, Guangzhou
		510641,  China}
	\author{Yan-Ling Wang}
	\email{wangylmath@yahoo.com}
	\affiliation{ School of Computer Science and Technology, Dongguan University of Technology, Dongguan, 523808, China}

\maketitle
\small
\begin{abstract}
 	 An orthogonal set of states in multipartite systems is called to be  strong quantum nonlocality  if it is locally irreducible under every bipartition of the subsystems [\href{https://journals.aps.org/prl/abstract/10.1103/PhysRevLett.122.040403}{Phys. Rev. Lett. \textbf{122}, 040403 (2019)}].
		  In this work, we study a subclass of locally irreducible sets: the only possible orthogonality preserving measurement on  each subsystems are   trivial  measurements.   We call the set with this property is    locally stable.   We find that in the case of two qubits systems locally stable sets \add{are} coincide with locally indistinguishable  sets.  Then we present a characterization of locally stable sets via the dimensions of some  states  depended \add{spaces.}  \add{Although the concept of locally stable set was proposed from the interest in mathematical properties, it also has its physical significance. One finds that   locally stable sets of orthogonal product states could not be perfectly distinguishable even with the use of asymptotic  local operations and classical
		  	communication (LOCC),
		  wherein an error is allowed but must vanish in the limit of an infinite number of rounds.} Moreover, we construct two   orthogonal sets in general multipartite quantum systems which are locally stable under every bipartition of the subsystems.   As a consequence, we obtain a lower bound and an upper bound  on the size of the smallest set which is locally stable for each bipartition of the subsystems.    Our results  provide  a complete answer to an open question (that is,  can we show
strong quantum nonlocality in       $\mathbb{C}^{d_1} \otimes \mathbb{C}^{d_1}\otimes \cdots \otimes \mathbb{C}^{d_N} $ for
any $d_i \geq 2$ and $1\leq i\leq   N$?) raised  in  a recent paper [\href{https://journals.aps.org/pra/abstract/10.1103/PhysRevA.105.022209}{Phys. Rev. A \textbf{105}, 022209 (2022)}].  Compared with all previous relevant proofs, our proof here is quite concise.
 
\end{abstract}

		\section{Introduction}  Quantum state  discrimination problem is a fundamental problem in  quantum information theory \cite{nils}.  A quantum system is prepared
		in a state which is randomly chosen from a known set.
		The task  is to identify the state of the system. If the known set is orthogonal,  then taking a states dependent projective  measurement can finish this task perfectly. However, if the known set is non-orthogonal, it is impossible to distinguish the states perfectly \cite{nils}.   Most of time,  our quantum states are   distributed in composite systems. So only local operations and classical
		communication (LOCC) are allowed in the distinguishing protocol.  Under this setting,  if the task can be  accomplished perfectly, we say that the set is  \emph{locally distinguishable},
		otherwise,  \emph{locally indistinguishable}.   Bennett et al. \cite{Ben99} presented the first  example
		of  orthogonal product states in $\mathbb{C}^3\otimes  \mathbb{C}^3$ that are locally indistinguishable and they named such a phenomenon as quantum nonlocality without entanglement.    The nonlocality here is in the sense that    the information of the given set that can be inferred by using global measurement is strictly larger than that obtained via    LOCC.     The results of  local  indistinguishability of  quantum states have  been practically applied in quantum cryptography primitives such as   data hiding \cite{Terhal01,DiVincenzo02} and secret sharing \cite{Markham08,Rahaman15,WangJ17}.

		Since Bennett et al.'s result \cite{Ben99}, the quantum nonlocality based on local discrimination    has been studied extensively (see Refs. \cite{Gho01,Wal00,Wal02,Fan04,Nathanson05,Cohen07,Bandyopadhyay11,Li15,Fan07,Yu12,Cos13,Yu115,Wang19,Xiong19,Li20,Banik21,Ran04,Hor03,Ben99b,DiVincenzo03,Zhang14,Zhang15,Zhang16,Xu16b,Xu16m,Zhang16b,Wang15,Feng09,		Yang13,Zhang17,Halder18,Li18,Halder1909,Xu20b,Rout1909}  for an
		incomplete list). \add{ There are two hot topics on the local discrimination of quantum states. One is aim at reducing the cardinality of the nonlocal sets. The other is to study some stronger form of nonlocal sets.}
		
		\add{ Adding orthogonal states to a  nonlocal set   forms a nonlocal set  again.  Therefore,  the smaller of the cardinality  means
			the better of the constructing nonlocal set. So it is interesting to ask  how small  a  nonlocal set could  be in some given systems? However,  the  answer is trivially being $3$ as
			we can always  embedding any  set with 3 elements from four Bell states (in two qubit systems)  into systems with large local dimensions  and more parties. There are many works   devoted to exploring small nonlocal sets of orthogonal states, all of which seem to have consciously avoided this trivial solution. One finds that all the proofs of  the local indistinguishability are based on the method derived by Walgate and   Hardy \cite{Wal02}. They observed that   in any locally distinguishable
			protocol one of the parties must go first, and whoever goes
			first must be able to perform some nontrivial orthogonality-preserving measurement (here a measurement $\{E_x\}_{x\in\mathcal{X}}$ is called nontrivial if not all the
			positive semidefinite operators  are proportional to the identity operator). Therefore, given an orthogonal set of multipartite systems, if one can show that the only possible  local orthogonal preserving measurements  for each partite are just the   trivial   measurements,  then one can conclude that the set is locally indistinguishable. Although this method has been applied for proving the local  indistinguishability of sets of quantum states  in  many research (see Refs. \cite{Zhang14,Zhang15,Wang15,Zhang16,Xu16b,Xu16m,Zhang16b,Zhang17,Halder18,Halder1909,Rout1909,Xu20b,Xu20a,Halder20c,Li21j}), its strength for proving   nonlocality   is still worth exploring. For example, we do not even know the minimum number of elements of the nonlocal set that can be derived by this method. From the perspective of mathematical research, it is both necessary and interesting to study the properties of nonlocal sets that can be described by this method.  This leads us to propose the concept of \emph{locally stable set}, an  orthogonal set  of multipartite quantum states  such that  the only possible
			orthogonality preserving measurement on each  subsystem are trivial  measurements. Under this defintion, it is interesting to find how small a locally stable set could be for a given multipartite systems. In this paper, we will provide some bounds on the cardinality of locally stable sets. Although arising from mathematical interest, locally stable sets are also found to have their physical meanings.  As a consequence of  unavoidable imperfections  in
			the real world, it is more appropriate to ask whether or not a task
			can be accomplished  with the amount
			of error arbitrary small.   If not,  the amount
			of error    is impossible to avoid. Recently, Cohen   \cite{Cohen22}   studied whether a set could be   perfectly  distinguishable under asymptotic LOCC,
			wherein an error is allowed but must vanish in the limit of an infinite number of rounds. Using a result from their work, we will find that locally stable sets of  orthogonal product states are  locally indistinguishable even in the sense of   asymptotic LOCC. }

		Recently, Halder \emph{et al.} \cite{Halder19} introduced  a  stronger form of local indistinguishability \add{which is} based on the concept of local irreducibility.   A set of multipartite orthogonal quantum states is said to be \emph{locally irreducible} if it is not possible to  eliminate one or more states from that set using orthogonality  preserving  local measurement.   A set of multipartite orthogonal quantum states is said to be \emph{strongly nonlocal} if it is locally irreducible for each bipartition of the subsystems. Although lots of study are focus on this topic (see  Refs.   \cite{Yuan20,Zhang1906,Shi20S,Wang21,Shi21,Shi21b,Shi22,Shi22N}), it remains several  open questions one of which is: does the   strong quantum nonlocality  exist  in       $\mathbb{C}^{d_1} \otimes \mathbb{C}^{d_2}\otimes \cdots \otimes \mathbb{C}^{d_N} $ for
any $d_i \geq 2$ and $1\leq i\leq   N$?  Similarly with the concept of strongly nonlocal, we should also study those sets of multipartite  orthogonal quantum states that are locally stable for each bipartition of the subsystems.        In Theorem \ref{thm:SN_MP},  we will show that  there do exist    sets  that are locally stable for each bipartition of the subsystems in       $\mathbb{C}^{d_1} \otimes \mathbb{C}^{d_1}\otimes \cdots \otimes \mathbb{C}^{d_N} $ for
any $d_i \geq 2$ and $1\leq i\leq   N$.  As locally  stable sets are always locally irreducible,   sets that are locally stable for each bipartition of the subsystems  are always strongly nonlocal. Therefore, our results  provide  a complete answer  to the aforementioned open question affirmatively.

		The rest of this article is organized as follows. In Sec. \ref{sec:concept}, we review the concept of locally indistinguishable set and introduce a special form of quantum nonlocality called locally stable set.   In Sec. \ref{sec:characterize}, we first give a complete descriptions of locally stable sets in two qubits systems. Then we give a characterization of locally stable sets by   some algebraic quantity known as the rank of some matrix.  In Sec. \ref{sec:twoconstruction}, we give two constructions of   sets that are locally stable for each bipartition of the subsytems in general multipartite quantum systems.     In Sec. \ref{sec:strongnonlocality},  we study the strong nonlocality of $W$ type states in multi-qubit systems.      Finally, we draw a conclusion and present some interesting problems in     Sec. \ref{sec:Conclusion}.

		\vskip 12pt	
		
		\section{Preliminaries}\label{sec:concept}

		For any positive integer $d\geq 2$, we denote $\bbZ_{d}$ as the set $\{0,1,\cdots,d-1\}$.	Let $\cH$ be an $d$ dimensional Hilbert space. We always assume that $\{|0\rangle, |1\rangle, \cdots,|d-1\rangle\}$ is the computational basis of $\cH$. A positive operator-valued measure (POVM) on $\mathcal{H}$ is a set of  positive semidefinite operators $\{E_x\}_{x\in\cX}$ such that $\sum_{x\in\cX} E_x= \mathbb{I}_{\mathcal{H}}$ where $\mathbb{I}_{\mathcal{H}}$  is the identity operator on   ${\mathcal{H}}$. \add{A measurement is  called trivial if   all its
		POVM elements  $\{E_x\}_{x\in\mathcal{X}}$   are proportional to the identity operator, i.e., $E_x\propto \mathbb{I}_\mathcal{H}$. }   For a finite set $\mathcal{S}$, we denote $|\mathcal{S}|$  as  the number of elements in that set. Throughout this paper, we do not normalize states
		for simplicity.
		
				 Now we give a brief review of some concept related to local discrimination of quantum states.
		
		\begin{definition}[Locally indistinguishable]\label{def: locally irreducible}
			A set of  orthogonal  pure states in multipartite quantum systems is said \emph{locally indistinguishable}, if it is not possible to distinguish the states by using LOCC.
			
		\end{definition}
		
		\begin{definition}[Locally irreducible]\label{def: locally irreducible}
			An orthogonal set of pure states in multipartite quantum systems is \emph{locally irreducible} if it is not
			possible to eliminate one or more states from the set by orthogonality-preserving local measurements.
		\end{definition}
	
	\add{	It has been pointed out that locally irreducible sets are always locally indistinguishable but the the opposite case is not true.}
  Note that every  LOCC protocol that distinguishes a set of orthogonal
	states is a sequence of orthogonality preserving local measurements (OPLM).   There is a sufficient condition to prove that an orthogonal set is   \add{locally indistinguishable or even  locally irreducible}: in each subsystem can only  perform  a trivial orthogonality preserving local measurement.    \add{ Given an orthogonal set $\mathcal{S}$ of pure states $\{|\Psi_i\rangle\}_{i=1}^n$ in multipartite systems $\otimes_{n=1}^N\mathcal{H}_{A_n}$, now we review a general method to show the local indistinguishability or local irreducibility of $\mathcal{S}$:  if the set could be locally distinguished or locally reducible, then at least one of the parties could start with a nontrivial orthogonality preserving measurement.     For example, if $A_n$ takes the first  nontrivial orthogonality preserving   measurement $\{E_x=M_x^\dagger M_x\}_{x}$ in the protocol, then at least one of $E_x$ is not proportional to $\mathbb{I}_{A_n}$ and the set of states  $\{M_x \otimes \mathbb{I}_{\hat{A}_n} |\Psi_i\rangle\}_{i=1}^n$ remain orthogonal for all $x$  where $\hat{A}_n:=\{A_1,A_2,\cdots,A_N\}\setminus\{A_n\},$ i.e.,
		\begin{equation}\label{eq:orthogonalRelation}
			\langle \Psi_j | E_x \otimes \mathbb{I}_{\hat{A}_n} |\Psi_i\rangle=0,  \  \text{for } 1\leq  i\neq j \leq n. 
		\end{equation}		   
		Therefore, if one can show that  $E_x\propto \mathbb{I}_{A_n}$ from Eqs.~\eqref{eq:orthogonalRelation} for all $n$, then one can conclude that the set is locally indistinguishable and locally irreducible. } This motivates us to introduce the following concept.
		
		\begin{definition}[Locally stable]\label{def: locally stable}
			An orthogonal set of pure states in multipartite quantum systems is said to be locally stable if the only possible
orthogonality preserving measurement on the subsystems are trivial  measurements. 
		\end{definition}
		
		\begin{figure}[h]
			\centering
			\includegraphics[scale=0.6]{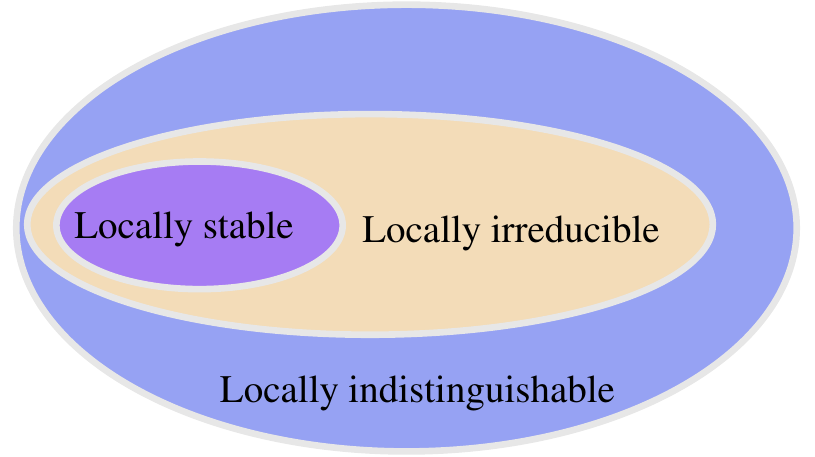}
			\caption{This is a schematic figure on the relations of the three concepts: locally indistinguishable, locally irreducible,  locally stable. }\label{fig:relation}
		\end{figure}
	\add{By the definition of locally stable set and the above proving strategy for a set to be locally irreducible, one finds that  locally stable  sets are always locally irreducible. Moreover, there exist locally irreducible sets that are not locally stable. In fact, the four Bell states 
		$$ \cS_B:=\{|\psi_{\pm}\rangle:=|00\rangle\pm|11\rangle, |\phi_\pm\rangle:=|01\rangle\pm|10\rangle\}$$ are    locally irreducible no matter looking them as  two qubit states or as two qutrit states.  However, they are not locally stable in   two qutrit systems $\mathbb{C}^3\otimes \mathbb{C}^3$ as $\{|0\rangle_A\langle 0|+|1\rangle_A \langle 1|, |2\rangle_A\langle 2|\}$ is a nontrivial orthogonality-preserving measurement for Alice.}   Therefore, locally stable sets present the strongest form of quantum nonlocality among the three classes: locally indistinguishable sets, locally irreducible sets and locally stable sets (See Figure \ref{fig:relation}). 
		
		There is some other form of stronger nonlocality considering the partition of the subsystems. In fact, a set of  orthogonal  pure states in multipartite quantum system is said to be  \emph{genuinely nonlocal} \cite{Rout1909,Li21j} if it is locally indistinguishable under each bipartition of the subsystems.    For the locally irreducible settings, Halder et al. \cite{Halder19} introduced the  \emph{strongly nonlocal set}  as the set   of  orthogonal  pure states in multipartite quantum system such that it is locally irreducible under each bipartition of the subsystems.   Therefore, in the setting of  locally stable, it is natural to study those sets   of  orthogonal  pure states in multipartite quantum systems such that they are locally stable  under each bipartition of the subsystems. In fact, sets with such property are called  \emph{strongest nonlocal sets} in Ref. \cite{Shi21}.  It is easy to deduce that strongest nonlocal sets are always strongly nonlocal and genuinely nonlocal. 
			
			In this paper, we mainly consider whether a given orthogonal set of  multipartite states is locally stable or  is of strongest nonlocality.   Let $\mathcal{H}=\otimes_{i=1}^N \mathcal{H}_{A_i}$ whose local dimensions are $\text{dim}_\mathbb{C} \mathcal{H}_{A_i}=d_i.$ We denote $s(d_1,d_2,\cdots,d_N)$ the smallest number of elements of those orthogonal sets in $\mathcal{H}$ that are locally stable. And   we denote $S(d_1,d_2,\cdots,d_N)$ the smallest number of elements  of those orthogonal sets in $\mathcal{H}$ that are of strongest nonlocality. In this work, we will give some bounds on the two quantities.

			\vskip 5pt
			\vskip 12pt	
			
			\section{Characterization of locally stable set and lower bounds on $s(d_1,d_2,\cdots,d_N)$ and $S(d_1,d_2,\cdots,d_N)$}\label{sec:characterize}
			First, we give a complete description of the locally stable  sets in two qubits systems.  The proof is very similar with that in  Ref. \cite{Wal02} where they considered locally indistinguishable sets.
			
			\begin{theorem}\label{thm:twoqubits_Case}
				Let $\cS=\{|\Psi_i\rangle\}_{i=1}^l$ be an orthogonal   set of  pure states in  $\mathbb{C}^2\otimes\mathbb{C}^2$.  Then
				$\cS$ is locally stable  if and only if   it is  locally indistinguishable. In fact, $\cS$ is locally stable if and only if  $|\cS|\geq3$ and contains at least two entangled states.
			\end{theorem}
			\begin{proof}
				If $\cS$ is locally stable, it is obviously local indistinguishable.
				If $\cS$ is not locally stable, then one of the parties may perform some nontrivial orthogonality preserving local measurement. Without loss of generality, we assume the first party can perform such a measurement. Hence there exists a semidefinite positive operator $E_x=M_x^\dagger M_x$ which is not proportional to $\mathbb{I}_2$ such that the elements in $\mathcal{S}':=\{M_x\otimes \mathbb{I}_2 |\Psi_i \rangle|i=1,\cdots,l\}$ are mutually orthogonal.  By the singular value decomposition, there are two   orthonormal sets  $\{|v_1\rangle, |v_2\rangle\}$  ($\langle v_i| v_j\rangle=\delta_{i,j} $)  and   $\{|w_1\rangle, |w_2\rangle\}$ ($\langle w_i| w_j\rangle =\delta_{i,j} $)  such that $$ M_x=\lambda_1|v_1\rangle\langle w_1|+\lambda_2|v_2\rangle\langle w_2|$$
				where $\lambda_1,\lambda_2\geq 0$ are  real numbers. As $E_x$   is not proportional to $\mathbb{I}_2$, we have $\lambda_1\neq \lambda_2.$ Each   $|\Psi_i\rangle$  ($1\leq i\leq l$) can be expressed as the following form
				$$|\Psi_i\rangle=|w_1\rangle|\psi_{i,1}\rangle+|w_2\rangle|\psi_{i,2}\rangle$$
				where $|\psi_{i,j}\rangle$ may be unnormalized and even be  a zero vector. As both $\mathcal{S}$ and $\mathcal{S}'$ are orthogonal sets, we have \begin{equation}
					\begin{array}{l}
						\langle \psi_{i,1}|\psi_{j,1}\rangle+  \langle \psi_{i,2}|\psi_{j,2}\rangle=0,\\[2mm]
						\lambda_1\langle \psi_{i,1}|\psi_{j,1}\rangle+ \lambda_2 \langle \psi_{i,2}|\psi_{j,2}\rangle=0
					\end{array}
				\end{equation}
				for $1\leq i\neq j\leq l.$ Therefore, we have $\langle \psi_{i,1}|\psi_{j,1}\rangle=\langle \psi_{i,2}|\psi_{j,2}\rangle=0$ as $\lambda_1\neq \lambda_2.$ With this at hand, Alice and Bob can provide a local protocol to distinguish the set $\mathcal{S}.$ In fact, Alice can perform the measurement $\pi_A:=\{|w_1\rangle\langle w_1|,|w_2\rangle\langle w_2|\}$ to the set $\mathcal{S}.$ For each outcome of this measurement, the possible states of Bob's party are orthogonal and hence can be distinguished by himself after   receiving the outcome of Alice.  
				
				The remaining result can be deduced from the known result (see Ref. \cite{Wal02})  that $\cS\subseteq \mathbb{C}^2\otimes\mathbb{C}^2$ is locally indistinguishable if and only if  $|\cS|\geq3$ and  $\mathcal{S}$ contains at least two entangled states. This completes the proof.
			\end{proof}

			\vskip 5pt

			Before studying the locally stable set for general bipartite systems, we  introduce some notation that may use throughout this section.

			Let $\mathcal{H}_1$ and $\mathcal{H}_2$ be two Hilbert spaces of dimensional $d_1$ and $d_2$ respectively. Denote
			$L(\mathcal{H}_2,\mathcal{H}_1)$ be all the linear operations mapping $\mathcal{H}_2$ to $\mathcal{H}_1$. There is a  linear one-one correspondence between the two linear spaces  $L(\mathcal{H}_2,\mathcal{H}_1)$ and $\mathcal{H}_1\otimes\mathcal{H}_2$. Exactly, this correspondence is given by the linear
			mapping  $\mathrm{vec}:L(\mathcal{H}_2,\mathcal{H}_1)  \rightarrow    {\mathcal{H}_1\otimes\mathcal{H}_2}$ defined by the action on the basis  $\mathrm{vec}(|i\rangle\langle j|):=|i\rangle |j\rangle.$ More generally, if $|a\rangle=\sum_{i\in\mathbb{Z}_{d_1}}a_i|i\rangle \in \mathcal{H}_1$ and  $|b\rangle=\sum_{j\in\mathbb{Z}_{d_2}}b_j|j\rangle\in\mathcal{H}_2$, then one can check that 
			\begin{equation}\label{eq:vecmap}
				\mathrm{vec}(|a\rangle\langle b|)=|a\rangle |\overline {b}\rangle
			\end{equation}
			where $|\overline{b}\rangle:=\sum_{j\in\mathbb{Z}_{d_2}}\overline{b_j}|j\rangle $ and $\overline{b_j}$  is the complex conjugate of ${b_j}.$
			This mapping is also an isometry, in the sense that
			\begin{equation}\label{eq: vecinner}
				\langle M, N  \rangle=\langle \mathrm{vec}(M), \mathrm{vec}(N)\rangle.
			\end{equation}   
			where $M, N\in  L(\mathcal{H}_2,\mathcal{H}_1).$ Here the inner product $ \langle M, N  \rangle$ is defined as $\mathrm{Tr}\big[M^\dagger N\big]$  for  $M, N\in  L(\mathcal{H}_2,\mathcal{H}_1)$ and $\langle |u\rangle, |v\rangle \rangle$ is defined as $\langle u|v\rangle$ for $ |u\rangle, |v\rangle \in \mathcal{H}_1\otimes\mathcal{H}_2.$

			Let $\mathcal{H}_{A}\otimes\mathcal{H}_B$ be a composed bipartite systems whose local dimensions  are $d_A$ and $d_B$ respectively. Suppose that $\{|i\rangle_{A}| i\in\mathbb{Z}_{d_A}\}$ and $\{|j\rangle_{B}| j\in\mathbb{Z}_{d_B}\}$ are the computational bases of  systems $A$ and $B$ respectively. Given an orthogonal set $\cS=\{|\psi_k\rangle\}_{k=1}^n$ of pure states in  $\mathcal{H}_{A}\otimes\mathcal{H}_B$, our goal is to determine whether there is a nontrivial orthogonality preserving local measurement to this set. 
			
			For each $|\psi_k\rangle$, we can express it as the form  $|\psi_k\rangle=\sum_{ i\in \mathbb{Z}_{d_A}}|i\rangle_{A}|\psi_{k,i}\rangle_{B} $
			where $|\psi_{k,i}\rangle_{B}$ may not be normalized and even may be equal to zero.
			If $E$ is a POVM element on subsystem $B$ that  preserves the orthogonality relation, then we have the   equalities $
			\langle \psi_k|\mathbb{I}_{A} \otimes E |\psi_l\rangle=0, \text{ for } k\neq  l.$
			Substituting the expressions of  $|\psi_k\rangle$  and $|\psi_l\rangle$  to  these equations, one obtain that
			\begin{equation*}\label{eq:relationdec}
				\sum_{i\in \mathbb{Z}_{d_A}} \tr{\Big[E |\psi_{k,i}\rangle_{B}\langle \psi_{l,i}|\Big]} =\sum_{i\in \mathbb{Z}_{d_A}} {}_{B}\langle \psi_{k,i}| E |\psi_{l,i}\rangle_{B}=0. 
			\end{equation*}
			Applying  Eqs. \eqref{eq:vecmap} and \eqref{eq: vecinner} to the left hand side  of the above equation, one obtains that  
			\begin{equation}\label{eq:relationvec}
				\langle \mathrm{vec}{(E)},\sum_{i\in \mathbb{Z}_{d_A}}    |\psi_{k,i}\rangle_{B}\otimes |\overline{\psi_{l,i}}\rangle_{B}\rangle=0
			\end{equation}
			whenever $k\neq l.$ That is, for each pair $(k,l)$ with $1\leq k\neq l \leq n$, the vector $\sum_{i\in \mathbb{Z}_{d_A}}    |\psi_{k,i}\rangle_{B}\otimes |\overline{\psi_{l,i}}\rangle_{B} $ is orthogonal to the vector $\mathrm{vec}{(E)}$ in the linear space $\mathcal{H}_B\otimes \mathcal{H}_B.$ Let $\mathcal{D}_{A|B}(\cS)$ denote the linear subspace of $\mathcal{H}_B\otimes \mathcal{H}_B$  spanned by $\sum_{i\in \mathbb{Z}_{d_A}}    |\psi_{k,i}\rangle_{B}\otimes |\overline{\psi_{l,i}}\rangle_{B} $ with  $1\leq k\neq l \leq n$. As one notes that $\mathrm{vec}(\mathbb{I}_B)$ is a nonzero vector that satisfies all the relations in Eq. \eqref{eq:relationvec}, we always have 
		\begin{equation}\label{eq:ineq}
		\mathrm{dim}_\mathbb{C}[\mathcal{D}_{A|B}(\cS)]\leq d_B^2-1.
		\end{equation}
			Moreover, as the map $\mathrm{vec}$ is an  isometry, if $ \mathrm{dim}_\mathbb{C}[\mathcal{D}_{A|B}(\cS)]= d_B^2-1$, one can conclude that the POVM measurement $M$ that satisfies all the relations in  Eq. \eqref{eq:relationvec} must be proportional to $\mathbb{I}_B.$  Similarly, we can define a subspace $\mathcal{D}_{B|A}(\cS)$ of  $\mathcal{H}_A\otimes \mathcal{H}_A$ in which case we use the decomposition of  $|\psi_k\rangle=\sum_{ j\in \mathbb{Z}_{d_B}}|j\rangle_{B}|\psi_{k,j}\rangle_{A}$
			where $|\psi_{k,j}\rangle_{A}$ may not be normalized and even may be equal to zero.

			\begin{theorem}\label{thm:AlgorithmBasis}
				Let $\cS$ be an orthogonal   set of  pure states in  $\mathcal{H}_{A}\otimes\mathcal{H}_B$   whose local dimensions  are $d_A$ and $d_B$ respectively.  Let $ \mathcal{D}_{A|B}(\cS)$ and $\mathcal{D}_{B|A}(\cS)$ be the linear spaces  defined as above. Then the set  $\cS$  is locally stable if and only if   both of the following equalities are satisfied $$\mathrm{dim}_\mathbb{C}[\mathcal{D}_{A|B}(\cS)]= d_B^2-1  \text{ and }  \mathrm{dim}_\mathbb{C}[\mathcal{D}_{B|A}(\cS)]= d_A^2-1.$$
			\end{theorem}
			\begin{proof} \emph{Sufficiency.}   If $\mathrm{dim}_\mathbb{C}[\mathcal{D}_{A|B}(\cS)]= d_B^2-1 $, by the previous statement, the Bob's site can only start with a trivial orthogonality preserving measurement. If $\mathrm{dim}_\mathbb{C}[\mathcal{D}_{B|A}(\cS)]= d_A^2-1$, so does Alice. Therefore, the set $\cS$ is locally stable by definition.
				
				\emph{Necessity.}    Suppose not, without loss of generality, we could assume that $\mathrm{dim}_\mathbb{C}[\mathcal{D}_{A|B}(\cS)]\leq d_B^2-2$ as by construction we always have $\mathrm{dim}_\mathbb{C}[\mathcal{D}_{A|B}(\cS)]\leq d_B^2-1.$  We define  $\cC_{A|B}(\cS)$ as the   $\mathbb{C}$-linear space  spanned by   $$\{\sum_{i\in \mathbb{Z}_{d_A}}    |\psi_{k,i}\rangle_{B}  \langle{\psi_{l,i}}|\ \big |\   1\leq k\neq l \leq n\}.$$  Because that the map $vec$ is an isometry,   the space   $\cC_{A|B}(\cS)$ has the same dimension as $\mathcal{D}_{A|B}(\cS).$ Note that for each $1\leq k\neq l\leq n$,  the set with two matrices
				$\{\sum_{i\in \mathbb{Z}_{d_A}}    |\psi_{k,i}\rangle_{B}  \langle{\psi_{l,i}}|,  \sum_{i\in \mathbb{Z}_{d_A}}    |\psi_{l,i}\rangle_{B}  \langle{\psi_{k,i}}|\}$ is linearly equivalent to the set with the following two  Hermitian matrices
				$$ \begin{array}{l}
					H_{kl}:= \displaystyle	\sum_{i\in \mathbb{Z}_{d_A}}    |\psi_{k,i}\rangle_{B}  \langle{\psi_{l,i}}|+|\psi_{l,i}\rangle_{B}  \langle{\psi_{k,i}}|, \\[2mm]
					H_{lk}:=\displaystyle  \sum_{i\in \mathbb{Z}_{d_A}} i|\psi_{k,i}\rangle_{B}  \langle{\psi_{l,i}}|-i|\psi_{l,i}\rangle_{B}  \langle{\psi_{k,i}}|. 
				\end{array}
				$$
				With this note, we have 
				$$\cC_{A|B}(\cS)=\text{span}_\mathbb{C}\{H_{kl} \ \big | \ 1\leq k\neq l\leq n\}.$$
				Now we define $\cR_{A|B}(\cS)=\text{span}_\mathbb{R}\{H_{kl} \ \big | \ 1\leq k\neq l\leq n\}.$
				We claim that 
				$$\text{dim}_\mathbb{R}[\cR_{A|B}(\cS)]\leq  \text{dim}_\mathbb{C}[\cC_{A|B}(\cS)].$$
				Clearly, $\cR_{A|B}(\cS)\subseteq \cC_{A|B}(\cS)$ and all the matrices in $\cR_{A|B}(\cS)$ are Hermitian. To prove the above claim, it is sufficient to prove that if  $H_1,H_2,\cdots,H_L\in \cR_{A|B}(\cS)$  are $\mathbb{R}$-linearly independent, then they are also  $\mathbb{C}$-linearly independent. If not, there exists not all zero $x_j+iy_j\in \mathbb{C}$, $1\leq j\leq L$ (here $x_j,y_j\in \mathbb{R}$ and we can always assume that some $x_j\neq 0$, otherwise multiplying both sides by the complex number $i$ )  such that 
				\begin{equation}\label{eq:Main_relation}
					\sum_{j=1}^L(x_j+iy_j)H_j=\boldsymbol{0}.
				\end{equation} 
				Taking the complex conjugacy to both sides, we obtain 
				\begin{equation}\label{eq:Main_relation_conjugate}
					\sum_{j=1}^L(x_j-iy_j)H_j=\boldsymbol{0}.
				\end{equation} 
				From Eqs. \eqref{eq:Main_relation} and \eqref{eq:Main_relation_conjugate}, we obtain that
				$\sum_{j=1}^L 2x_j H_j=\boldsymbol{0}$ which is contradicted with the assumption that $H_1,H_2,\cdots,H_L$  are $\mathbb{R}$-linearly independent.  Therefore, we have 
				$$\text{dim}_\mathbb{R}[\cR_{A|B}(\cS)]\leq  \text{dim}_\mathbb{C}[\cC_{A|B}(\cS)]\leq d_B^2-2.$$
				We know that all the $d_B\times d_B$ Hermitian matrices form an $\mathbb{R}$-linear space of dimensional $d_B^2$. As $\mathbb{I}_{d_B}$ lies in the completion space of $\cR_{A|B}(\cS)$ and $\text{dim}_\mathbb{R}[\cR_{A|B}(\cS)]\leq   d_B^2-2$, there exists at least some other nonzero Hermitian matrix said $E_B$ which is orthogonal to the space $\cR_{A|B}(\cS)$ and the identity  matrix $\mathbb{I}_{d_B}$. Multiplying some nonzero real number, we can always assume that each   eigenvalue   $\lambda_j$ of $E_B$  satisfies $|\lambda_j|\leq 1/4$. Then we have both $E_1:=\mathbb{I}_{d_B}/2+E_B$ and $E_2:=\mathbb{I}_{d_B}/2-E_B$ are semidefinite positive and $E_1+E_2=\mathbb{I}_{d_B}$. Therefore, $\{E_1, E_2\}$ is a POVM. By definition, $E_1, E_2$ lie in the  completion space of $ \cR_{A|B}(\cS)$, hence it is an orthogonality preserving measurement with respect to the set $\cS$.  Moreover, it is easy to see that it is a nontrivial measurement. So this is contradicted to the condition that $\cS$ is locally stable. Therefore, we should  have $\mathrm{dim}_\mathbb{C}[\mathcal{D}_{A|B}(\cS)]\geq d_B^2-1.$ Combining this with Eq. \eqref{eq:ineq}, we deduce $\mathrm{dim}_\mathbb{C}[\mathcal{D}_{A|B}(\cS)]=d_B^2-1.$  This completes the proof.
			\end{proof}	
			
			As the subspace $\mathcal{D}_{B|A}(\cS)$  (w.r.t. $\mathcal{D}_{A|B}(\cS)$) is completely determined by  $(n-1)n$  generators, therefore,  we can use an $(n-1)n\times d_A^2$ (w.r.t.  $(n-1)n\times d_B^2$)  matrix to represent it. And we denote the matrix as $D_{B|A}(\cS)$ (w.r.t. $D_{A|B}(\cS)$).
			\begin{example}\label{ex:Bell}
				The set $\cS_B$ with three Bell states is locally stable, 
				where 
				$$ \cS_B:=\{|\psi_{\pm}\rangle:=|00\rangle\pm|11\rangle, |\phi_+\rangle:=|01\rangle+|10\rangle\}.$$
			\end{example}
			One can easily calculate out the matrix   $D_{A|B}(\cS_B)$ as follows
			$$
			D_{A|B}(\cS_B)=\left[
			\begin{array}{rrrr}
				1& 0 &0 &-1\\
				0& 1& 1& 0\\
				1& 0 &0 &-1\\
				0& 1& -1& 0\\
				0& 1& 1& 0\\
				0& -1& 1& 0\\
			\end{array}
			\right].
			$$
			For example, consider the pair $(|\psi_+\rangle,|\phi_+\rangle)$ of  states. We have  $|\psi_+\rangle=|0\rangle|\psi_0\rangle+|1\rangle|\psi_1\rangle$ and $|\phi_+\rangle=|0\rangle|\phi_0\rangle+|1\rangle|\phi_1\rangle$ where $|\psi_0\rangle=|\phi_1\rangle=|0\rangle $ and  $|\psi_1\rangle=|\phi_0\rangle=|1\rangle.$ It contributes the vector $|\phi_0\rangle|\bar{\psi}_0\rangle+|\phi_1\rangle|\bar{\psi}_1\rangle=|1\rangle|0\rangle+|0\rangle|1\rangle$, i.e., the vector $[0, 1, 1, 0]$ under the computational basis $\{|00\rangle,|01\rangle,|10\rangle,|11\rangle\}.$  
			And $\mathrm{rank}(D_{A|B}(\cS_B))=3.$ By symmetry, we can conclude that $\cS_B$ is locally stable.
			
			\vskip 5pt

		  It is not difficult to generalize  the results of Theorem \ref{thm:AlgorithmBasis} to multipartite systems.

			\begin{theorem}
				\label{thm:AlgorithmBasis_Strong}
				Let $\cS$ be an orthogonal   set of  pure states in  $\otimes_{i=1}^N\mathcal{H}_{A_i}$   whose local dimension  is  $\mathrm{dim}_\mathbb{C}(\mathcal{H}_{A_i})=d_{i}$. Then we have the following statements:
				
				\begin{enumerate}[(a)]
					\item  The set  $\cS$  is  locally  stable 
					if and only if   all the equalities $\mathrm{dim}_\mathbb{C}[\mathcal{D}_{\hat{A_i}|A_i}(\cS)]= d_i ^2-1   ( i=1, \cdots,N)$ are satisfied where we use the  notation \begin{equation}\label{eq:hat}
					    \hat{A}_i:= \{A_1,A_2,\cdots,A_N\}\setminus\{A_i\}
					\end{equation}
					i.e., the union of all subsystems except the $A_i$.
					\item  The set  $\cS$  is  of strongest nonlocality if  and only if  all the equalities $\mathrm{dim}_\mathbb{C}[\mathcal{D}_{A_i|\hat{A_i}}(\cS)]= \hat{d_i}^2-1   ( i=1, \cdots,N )$ are satisfied where we  denote \begin{equation}\label{eq:hatd}\hat{d}_i=(\prod_{j=1}^N d_j)/d_i.	\end{equation}
					
				\end{enumerate}
			\end{theorem}
			\begin{proof}
				The proof is similar with the proof in Theorem \ref{thm:AlgorithmBasis}. In fact, the essence in the  proof  of Theorem \ref{thm:AlgorithmBasis} shows that $\mathrm{dim}_\mathbb{C}[\mathcal{D}_{\hat{A_i}|A_i}(\cS)]= d_i ^2-1$ if and only if the $A_i$ party can only perform a trivial orthogonality preserving measurement. The same reason that  $\mathrm{dim}_\mathbb{C}[\mathcal{D}_{A_i|\hat{A}_{i}}(\cS)]= \hat{d}_i ^2-1$ if and only if the $\hat{A}_i$ party can only perform a trivial orthogonality preserving measurement. With these two equivalent relations, it is easy to recover the above two statements.
			\end{proof}
	\vskip 5pt 
 
	\add{	For a given orthogonal set of product states $\mathcal{S}=\{\otimes_{n=1}^N |\psi^{(n)}_i \rangle\}$ in $\otimes_{i=1}^N\mathcal{H}_{A_i}$.	 Define  $J_n$ to be the set
		$$\{(i,j)\mid \langle \psi^{(n)}_i| \psi^{(n)}_j\rangle  =0;  \langle \psi^{(m)}_i| \psi^{(m)}_j\rangle  \neq 0\  \forall m\neq n\}.$$ The Theorem 3 of Ref.~ \cite{Cohen22} shows that if for every party $n$, the set $\mathcal{L}_n:=\{|\psi_i^{(n)} \rangle\langle  \psi^{(n)}_j| \ \big | (i,j)\in J_n\} $ spans a space of dimension $d_n^2
		-1$, then the set $\mathcal{S}$  cannot be perfectly discriminated under asymptotic LOCC. Using this result and noting that $\mathcal{D}_{\hat{A}_n|A_n}(\mathcal{S})=\mathcal{L}_n$, one could easily deduce the following corollary by the statement (a) of  Theorem  \ref{thm:AlgorithmBasis_Strong}.}
			
	\begin{corollary}
	\add{	 Let $\cS$ be an orthogonal   set of  pure states in  $\otimes_{i=1}^N\mathcal{H}_{A_i}$   whose local dimension  is  $\mathrm{dim}_\mathbb{C}(\mathcal{H}_{A_i})=d_{i}$. If $\cS$ is locally stable, then  the perfect discrimination is impossible   by asymptotic LOCC  wherein an error is allowed but must vanish in the limit of an infinite number of rounds.}
	\end{corollary}

			Using Theorem  \ref{thm:AlgorithmBasis_Strong}, we could  derive a bound on the cardinality of a locally stable set in  multipartite systems. 
			
			\begin{theorem}\label{thm:AlgorithmBasis_Bounds}  $($\emph{\textbf{Bounds on the sizes of  locally stable sets}}$)$  
				Let $\cS$ be an orthogonal   set of  pure states in  $\otimes_{i=1}^N\mathcal{H}_{A_i}$   whose local dimension  is  $\text{dim}_\mathbb{C}(\mathcal{H}_{A_i})=d_{i}$.The we have the following statement:
				
				\begin{enumerate}[(a)]
					\item 
					If  the set  $\cS$  is  locally  stable, then  
					$|\cS| \geq \max_i\{d_i+1\}.$ Consequently,   $$s(d_1,d_2,\cdots,d_N)\geq \max_i\{d_i+1\}.$$
					\item  If    the set  $\cS$  is of strongest nonlocality,  then 	$$|\cS| \geq \max_i\{\hat {d}_i+1\}$$ (see Eq. \eqref{eq:hatd} for definition of $\hat {d}_i$). Consequently,   $$S(d_1,d_2,\cdots,d_N)\geq \max_i\{\hat {d}_i+1\}.$$
				\end{enumerate}
				
			\end{theorem}
			\noindent\emph{Proof.}  (a)
			By Theorem \ref{thm:AlgorithmBasis_Bounds}, we have  $\mathcal{D}_{\hat{A_i}|A_i}(\cS)=d_i^2-1$. And by the definition of $\mathcal{D}_{\hat{A_i}|A_i}(\cS)$, we have   $$d_i^2-1=\mathrm{dim}_\mathbb{C}[\mathcal{D}_{\hat{A_i}|A_i}(\cS)]\leq|\cS|^2-|\cS|.$$ Therefore, $|\cS|>d_i$ for each $1\leq i\leq N$.
			
			(b)The proof is similar with the above  proof. \qed
			
			\vskip 12pt

			
			\vskip 12pt

			\section{Two constructions of strongest nonlocal  sets and upper bounds on  $S(d_1,d_2,\cdots,d_N)$}\label{sec:twoconstruction}

			Generally, it is difficult to show that the possible orthogonality preserving local measurement for each subsystem is trivial. We list two useful     lemmas (developed in Ref.    \cite{Shi21})  for verifying  the trivialization of such measurement.

			\begin{lemma}[Block Zeros Lemma]\label{lem:zero}
				Let  an  $n\times n$ matrix $E=(a_{i,j})_{i,j\in\bbZ_n}$ be the matrix representation of an operator  $E$  under the basis  $\cB:=\{\ket{0},\ket{1},\ldots,\ket{n-1}\}$. Given two nonempty disjoint subsets $\cS$ and $\cT$ of $\cB$, assume  that  $\{\ket{\psi_i}\}_{i=0}^{s-1}$, $\{\ket{\phi_j}\}_{j=0}^{t-1}$ are two orthogonal sets  spanned by $\cS$ and $\cT$ respectively, where $s=|\cS|,$ and $t=|\cT|.$  If  $\langle \psi_i| E| \phi_j\rangle =0$
				for any $i\in \mathbb{Z}_s,j\in\mathbb{Z}_t,$ then   $\langle x|E|y\rangle=\langle y|E|x\rangle=0$ for $|x\rangle\in \mathcal{S}$ and $|y\rangle\in \mathcal{T}.$
			\end{lemma}

			\begin{lemma}[Block Trivial  Lemma]\label{lem:trivial}
				Let  an  $n\times n$ matrix $E=(a_{i,j})_{i,j\in\bbZ_n}$ be the matrix representation of an operator  $E$  under the basis  $\cB:=\{\ket{0},\ket{1},\ldots,\ket{n-1}\}$. Given a nonempty  subset $\cS$  of $\cB$, let $\{\ket{\psi_j} \}_{j=0}^{s-1}$ be an orthogonal  set spanned by $\cS$.     Assume that $\langle \psi_i|E |\psi_j\rangle=0$ for any $i\neq j\in \mathbb{Z}_s$.  If there exists a state $|x\rangle \in\cS$,  such that $ \langle x|E |y\rangle =0$ for all $|y\rangle \in \mathcal{S}\setminus \{|x\rangle\}$  and $\langle x|\psi_j\rangle \neq 0$  for any $j\in \mathbb{Z}_s$, then  $\langle y|E|z\rangle=0$ and $\langle y|E|y\rangle=\langle z|E|z\rangle$ for all $|y\rangle,|z\rangle\in \mathcal{S}$ with $|y\rangle\neq |z\rangle.$
			\end{lemma}

			In this section, we provide two constructions of strongest nonlocal  sets: $\mathcal{S}$ (all but one state are genuinely entangled) and  $\mathcal{S}_G$ (all    states  are genuinely entangled).
			
			Let $\mathcal{H}:=\mathcal{H}_{A_1}\otimes\mathcal{H}_{A_2} \otimes\cdots \otimes\mathcal{H}_{A_N}$ be an $N$ parties quantum systems with dimensional $d_i$ for the $i$th subsystem.
			A string $\boldsymbol{i}=(i_1,i_2,\cdots,i_N)$ in  $\boldsymbol{C}:=\mathbb{Z}_{d_1}\times\mathbb{Z}_{d_2}\times\cdots\times\mathbb{Z}_{d_N}$ is called weight $k$ if there are exactly $k$ nonzero $i_j$'s. And we denote the set of all weight $k$ strings of $\boldsymbol{C}$ as $\boldsymbol{C}_k$ where $0\leq k\leq N$. Set $c_k:=|\boldsymbol{C}_k|,$ i.e., the number of elements in   $\boldsymbol{C}_k$. For each $k\in\{0,1,\cdots, N-1\}$, we define 
			
			$$\mathcal{S}_k:=\{|\Psi_{k,i}\rangle\in \mathcal{H}\  \big|\  i\in\mathbb{Z}_{c_k}, |\Psi_{k,i}\rangle:=\sum_{\boldsymbol{j}\in \boldsymbol{C}_k}\omega_{c_k}^{if_k({\boldsymbol{j}})}|\boldsymbol{j}\rangle\}.$$
			Here $f_k: \boldsymbol{C}_k\rightarrow \mathbb{Z}_{c_k}$   is  any fixed  bijection  and $\omega_n:=e^{\frac{2\pi \sqrt{-1}}{n}}.$

			\begin{theorem}\label{thm:SN_MP}
				Let $\mathcal{H}:=\mathcal{H}_{A_1}\otimes\mathcal{H}_{A_2} \otimes\cdots \otimes\mathcal{H}_{A_N}$ be an $N$ parties quantum systems with dimensional $d_i$ for the $i$th subsystem.  The set $\mathcal{S}:=\cup_{k=0}^{N-1} \mathcal{S}_k$ is of strongest   nonlocality. Then  $|\mathcal{S} |= \prod_{n=1}^Nd_n-\prod_{n=1}^N(d_n-1).$ As a consequence, 
				$$ S(d_1,d_2,\cdots,d_N)\leq \prod_{n=1}^Nd_n-\prod_{n=1}^N(d_n-1).$$
			\end{theorem}
			\begin{proof}
				First, we show that $\hat{A}_1:=A_2A_3\cdots A_N$ can only perform a trivial orthogonality preserving measurement (OPM). Suppose that $\{M_x^\dagger M_x\}_{x\in\mathcal{X}}$ is an  orthogonality preserving measurement with respect to the set $\mathcal{S}$ which is performed by $\hat{A}_1$, i.e., 
				$\langle \Psi|\mathbb{I}_{A_1}\otimes M_x^\dagger M_x |\Phi\rangle=0$
				for any two different $|\Psi\rangle,|\Phi\rangle\in \mathcal{S}$. Set $E:=\mathbb{I}_{A_1}\otimes M_x^\dagger M_x$. Let $k,l\in \mathbb{Z}_N$ and suppose $k\neq l$. As $\boldsymbol{C}_k\cap \boldsymbol{C}_l=\emptyset$, applying {\bf Block Zeros Lemma} to the sets of base vectors corresponding to $\boldsymbol{C}_k$ and $\boldsymbol{C}_l$, we obtain that
				$$ \langle  \boldsymbol{i}_k |E|\boldsymbol{j}_l\rangle=\langle  \boldsymbol{j}_l |E|\boldsymbol{i}_k\rangle=0$$
				for any $\boldsymbol{i}_k\in \boldsymbol{C}_k$ and  $\boldsymbol{j}_l\in \boldsymbol{C}_l$. 
				Now we claim that for any  $k\in \{0,1,\cdots,N-1\}$ if $\boldsymbol{i}_k, \boldsymbol{i}_k'$ are two different strings of $\boldsymbol{C}_k$, then we also have 
				$$ \langle  \boldsymbol{i}_k |E|\boldsymbol{i}_k'\rangle=\langle  \boldsymbol{i}_k' |E|\boldsymbol{i}_k\rangle=0.$$
				Moreover, $\langle  \boldsymbol{i}_k |E|\boldsymbol{i}_k\rangle=\langle  \boldsymbol{0} |E|\boldsymbol{0}\rangle.$	   	 
				As $\boldsymbol{C}_{<N}:=\cup_{k<N} \boldsymbol{C}_k$	    contains $\{0\}\times\mathbb{Z}_{d_2}\times\cdots\times \mathbb{Z}_{d_N} $, from the above relations, one could conclude that $M_x^\dagger M_x\propto \mathbb{I}_{\hat{A}_1}.$ 
				
				In the following, we will give a proof of the above claim by induction. First, the claim is true for $k=0.$ Now we assume that this claim is true for $0\leq k<N-1$. Let $l=k+1$ and fix any  $\boldsymbol{j}_l=(j_1,j_2,\cdots,j_N)\in \boldsymbol{C}_l$ such that $j_1\neq 0$. For any $\boldsymbol{i}_l=(i_1,i_2,\cdots,i_N)\in \boldsymbol{C}_l$ which is different from $\boldsymbol{j}_l$. If $i_1\neq j_1$, 
				$$\langle  \boldsymbol{j}_l |E|\boldsymbol{i}_l\rangle=\langle  \boldsymbol{j}_l |\mathbb{I}_{A_1}\otimes M_x^\dagger M_x|\boldsymbol{i}_l\rangle=0.$$ If $i_1= j_1$,	set $\boldsymbol{j}_k:=(0,j_2,\cdots,j_N)$ and $\boldsymbol{i}_k:=(0,i_2,\cdots,i_N)$, by definition,  they are different strings of $\boldsymbol{C}_k$. Moreover, 
				$$\langle  \boldsymbol{j}_l |E|\boldsymbol{i}_l\rangle=\langle  j_2\cdots j_N|M_x^\dagger M_x|i_2\cdots i_N \rangle=\langle  \boldsymbol{j}_k |E|\boldsymbol{i}_k\rangle=0$$
				by induction. Applying {\bf Block Trivial  Lemma} to the set of base vectors corresponding to $\boldsymbol{C}_l$, the set $\{|\Psi_{l,i}\rangle\}_{i\in \mathbb{Z}_{c_l}}$ and  the vector $|\boldsymbol{j}_l\rangle$, we obtain that
				for any different strings $\boldsymbol{i}_l,\boldsymbol{i}_l'\in \boldsymbol{C}_l$,
				$$ \langle  \boldsymbol{i}_l |E|\boldsymbol{i}_l'\rangle=\langle  \boldsymbol{i}_l' |E|\boldsymbol{i}_l\rangle=0,  \ \ \text{ and }  \langle  \boldsymbol{i}_l |E|\boldsymbol{i}_l\rangle=\langle  \boldsymbol{j}_l |E|\boldsymbol{j}_l\rangle.$$
				Note that $\langle  \boldsymbol{j}_l |E|\boldsymbol{j}_l\rangle$ equals to 
				$$\langle  j_2\cdots j_N|M_x^\dagger M_x|j_2\cdots j_N \rangle=\langle  \boldsymbol{j}_k |E|\boldsymbol{j}_k\rangle=\langle  \boldsymbol{0} |E|\boldsymbol{0}\rangle. $$
				
				This completes the proof of the claim.  Therefore, the last $(N-1)$-parties could only start with a trivial OPM.
				
				By the  symmetric construction, one can also show that any $(N-1)$ parties could only start with a trivial OPM. This statement also implies that any $k$ (where $1\leq k\leq N-1)$ parties could only start with a trivial OPM. 
				
			\end{proof}
			
			Note that the elements in $\mathcal{S}$ are not always with genuine entanglement. In fact, $|\Psi_0\rangle=|{\boldsymbol{0}}\rangle$ is fully product states. Now we claim that except this state, all others are with genuine entanglement. We only need to show that $|\Psi_{k,i}\rangle$ $(1\leq k\leq N-1, i\in \mathbb{Z}_{c_k})$ is entangled for any bipartition  of the subsystems.  We assume that the bipartition is $\{A_i|i\in \mathcal{I}\} | \{A_j|j\in \mathcal{J}\}$ where $\mathcal{I}, \mathcal{J}$ are nonempty subsets of $\{1,2,\cdots,N\}$,  disjoint and $\mathcal{I}\cup\mathcal{J}=\{1,2,\cdots,N\}$.   Let $\mathcal{A}$ and $\mathcal{B}$ denote the computational bases of the systems $\{A_i|i\in \mathcal{I}\} $ and $ \{A_j|j\in \mathcal{J}\}$ respectively. Suppose that $|\Psi_{k,i}\rangle=\sum_{|a\rangle\in \mathcal{A}} \sum_{|b\rangle\in \mathcal{B}} \psi_{a,b}|a\rangle|b\rangle.$ It sufficient to prove that the rank of the matrix $(\psi_{a,b})$ is greater than one.  Clearly, $k$ can be expressed as two different  forms $k=s+t$ such that $0\leq s\leq |\mathcal{I}|$ and  $0\leq t\leq |\mathcal{J}|.$ Suppose $k=s_1+t_1=s_2+t_2$ such that  $0\leq s_1<s_2\leq |\mathcal{I}|$  and  $0\leq 
			t_2<t_1\leq |\mathcal{J}|.$  Choose any subsets $\mathcal{I}_x\subset \mathcal{I}$ ($\mathcal{J}_y\subset \mathcal{J}$) such that $|\mathcal{I}_x|=s_x$ for $x=1,2$ ($|\mathcal{J}_y|=t_y$  for $y=1,2$). We define
			$$
			\begin{array}{rl}
				|\mathcal{I},\mathcal{I}_x\rangle:=&\ \left((\otimes_{i\in \mathcal{I}_x}  |1\rangle_{A_i})\otimes  (\otimes_{ i\in \hat{\mathcal{I}}_x }  |0\rangle_{A_i}) \right) \in \mathcal{A}, \\
				|\mathcal{J},\mathcal{J}_y\rangle:=&\left((\otimes_{j\in \mathcal{J}_y}  |1\rangle_{A_j})\otimes  (\otimes_{ j\in \hat{\mathcal{J}}_y }  |0\rangle_{A_j})\right) \in \mathcal{B},  \\
			\end{array}
			$$
			where $x,y\in \{1,2\}.$
			The matrix $(\psi_{a,b})$ has the   $2\times 2$ minor
			$$
			\begin{array}{ccc}
				& |\mathcal{J},\mathcal{J}_2\rangle & |\mathcal{J},\mathcal{J}_1\rangle\\
				|\mathcal{I},\mathcal{I}_1\rangle & 0 & \alpha\\
				|\mathcal{I},\mathcal{I}_2\rangle & \beta &0
			\end{array}
			$$
			where $\alpha\beta \neq 0.$ Therefore, the Schmidt rank of $|\Psi_{k,i}\rangle$ across this partition  $\{A_i|i\in \mathcal{I}\} | \{A_j|j\in \mathcal{J}\}$ is greater than 1. Hence it is entangled.
			
			In Theorem \ref{thm:SN_MP}, if we replace the set $\mathcal{S}_0$ by two states $|\Psi_{\pm}\rangle:=|{\boldsymbol{0}}\rangle\pm |{\boldsymbol{1}}\rangle $ where $|\boldsymbol{0}\rangle=\otimes_{i=1}^N|0\rangle_{A_i},|\boldsymbol{1}\rangle=\otimes_{i=1}^N|1\rangle_{A_i}$ and  denote the new  total set as $\mathcal{S}_G$, then the set $\mathcal{S}_G$  is genuinely entangled set that also has property of strong nonlocality. In fact, for each $1\leq k\leq N-1, i\in \mathbb{Z}_{c_k}$, from  the orthogonal relations 
			$$\langle \Psi_\pm |E |\Psi_{k,i}\rangle=0,$$
			we can deduce the orthogonal relations 
			$\langle {\boldsymbol{0}} | E |\Psi_{k,i}\rangle=0.$
			Therefore, the orthogonal relations of  $\mathcal{S}_G$ contains those from $\mathcal{S}$ in Theorem \ref{thm:SN_MP}. Using these relations, we could obtain that  $\mathcal{S}_G$ is also of strongest nonlocality.
			\vskip 8pt

\section{More examples with strongest  nonlocality}\label{sec:strongnonlocality}In this part, we try to use the algebraic quantities  in section \ref{sec:characterize} to find   more sets which  has the property  of strongest nonlocality. The first result is out of our expectation. Using entangled states, three qubits is enough to show the strongest  nonlocality.
We use the following Pauli gate operations 
$X,Y, Z$   
{\small$$
	X=\left[
	\begin{array}{cc}
		0& 1\\
		1& 0
	\end{array}
	\right], \ \
	Y=\left[
	\begin{array}{cc}
		0& i\\
		-i& 0
	\end{array}
	\right],\ \ Z=\left[
	\begin{array}{cc}
		1&0\\
		0& -1
	\end{array}
	\right], 
	$$}
and the identity operation $ I$.
In $\mathbb{C}^2\otimes \mathbb{C}^2\otimes \mathbb{C}^2$, the $W$ state is $|100\rangle+|010\rangle+|001\rangle$. Let $U^{(W)}$ denote the matrix $$X\otimes I\otimes I+Z\otimes X\otimes I+Z\otimes Z\otimes X.$$
Then the set $\cS_W:=\{U^{(W)}|ijk\rangle \ \mid\  i,j,k\in \mathbb{Z}_2 \}$ (the states can be seen  in the following table) is an orthogonal  basis whose elements are  all locally unitary equivalent to the above $W$ state \cite{Miyake05}. 

\begin{table}[h]
	\centering	
	$\begin{array}{ccc}\hline\hline
		|ijk\rangle &   U^{(W)}|ijk\rangle & \text{ label }\\ \hline
		|000\rangle &|100\rangle+|010\rangle+|001\rangle & |\psi_1\rangle\\
		|001\rangle &|101\rangle+|011\rangle+|000\rangle & |\psi_2\rangle\\
		|010\rangle &|110\rangle+|000\rangle-|011\rangle & |\psi_3\rangle\\
		|011\rangle &|111\rangle+|001\rangle-|010\rangle & |\psi_4\rangle\\
		|100\rangle &|000\rangle-|110\rangle-|101\rangle & |\psi_5\rangle\\
		|101\rangle &|001\rangle-|111\rangle-|100\rangle & |\psi_6\rangle\\
		|110\rangle &|010\rangle-|100\rangle+|111\rangle & |\psi_7\rangle\\
		|111\rangle &|011\rangle-|101\rangle+|110\rangle & |\psi_8\rangle\\ \hline\hline
	\end{array}
	$
\end{table}

Using Matlab we can show that the  ranks of the matrices $D_{R|\hat{R}}(\cS_W)$  ($R=A,B,C$)  are  all 15.  Therefore, by Theorem \ref{thm:AlgorithmBasis_Strong}, the set $\cS_W$ is of strongest   nonlocality. Moreover, the statement holds also for  any its subset with 6 elements. 

Generally, we have  an orthonormal basis whose elements are all  locally unitary equivalent to the generalized W state in $N$-qubit \cite{Miyake05}.    For any integer $N\geq 3$,  we denote    $$U_N^{(W)}:=\sum_{l=1}^N Z_1\cdots Z_{l-1}X_lI_{l+1}\cdots I_N.$$ The set defined by  $\cS^{(W)}_N:=\{U_N^{(W)}|\boldsymbol{i}\rangle \ \mid \  \boldsymbol{i}\in \mathbb{Z}_2^N\}$ is such a set.
\begin{lemma}\label{lem:W_basis}
	The set $\cS^{(W)}_N$ is   an orthogonal set of pure states whose elements are all locally unitary equivalent to the multiqubit $W$ state $|W_N\rangle:=\sum_{l=1}^N|0_1\cdots0_{l-1}1_l0_{l+1}\cdots 0_N\rangle$ (see Figure \ref{fig:Wbasiscurcuit}). 
\end{lemma}

\begin{proof}
	The prove that $\cS^{(W)}_N$ is an orthogonal set, it is sufficient to prove  that $U_N^{(W)}$ is a unitary matrix up to a constant. This can be easily deduced when we notice that the anti-commutative relation $XZ=-ZX$. Therefore,
\begin{widetext}
    	{ $$	\begin{array}{rl}
			&U_N^{(W)}{U_N^{(W)}}^\dagger  \\[2mm] =&\displaystyle\sum_{k=1}^N\sum_{l=1}^N ( Z_1\cdots Z_{k-1}X_kI_{k+1}\cdots I_N)( Z_1\cdots Z_{l-1}X_lI_{l+1}\cdots I_N)\\[4mm]
			=&  \displaystyle\sum_{k<l} I_1\cdots I_{k-1}(X_kZ_k+Z_kX_k)Z_{k+1}\cdots Z_{l-1} X_l I_{l+1}\cdots I_N +\displaystyle\sum_{k=1}^N   I_1\cdots I_{k-1}I_kI_{k+1}\cdots I_N\\[4mm]
			=&N I^{\otimes N}.
		\end{array}
		$$
	}
	
	Clearly,  $U_N^{(W)}|\boldsymbol{0}\rangle =|W_N\rangle$ which is exactly the $N$-qubit $W$ state. For any $\boldsymbol{i}=(i_1,i_2,\cdots, i_N) \in \mathbb{Z}_2^N$, the state $|\psi_{\boldsymbol{i}}\rangle:=U_N^{(W)}|\boldsymbol{i}\rangle$
	can be written as 
	$$|\psi_{\boldsymbol{i}}\rangle= U_N^{(W)} X_1^{i_1}X_2^{i_2}\cdots X_N^{i_N} |\boldsymbol{0}\rangle. $$
	One can check that
	{ 	$$  X_1^{i_1} X_2^{i_2} \cdots X_N^{i_N} |\psi_{\boldsymbol{i}}\rangle=\sum_{l=1}^N (-1)^{i_1+\cdots+ i_{l-1}}|0_1\cdots0_{l-1}1_l0_{l+1}\cdots 0_N\rangle.$$
	}
\end{widetext}
	For each $l\in \{1,\cdots, N\}$, define $\theta_l:=i_1+\cdots+ i_{l-1}$ and 
	$Z(\theta):=|0\rangle\langle 0|+(-1)^\theta|1\rangle\langle 1|.$ Then 
	$$ 
	\otimes_{l=1}^N(Z_l(\theta_l)X_l^{i_l}) |\psi_{\boldsymbol{i}}\rangle=|W_N\rangle.
	$$
	That is, $|\psi_{\boldsymbol{i}}\rangle$ is locally unitary equivalent to $|W_N\rangle$.  
\end{proof}

\begin{figure}[ht]
	\centering
	\includegraphics[scale=0.55]{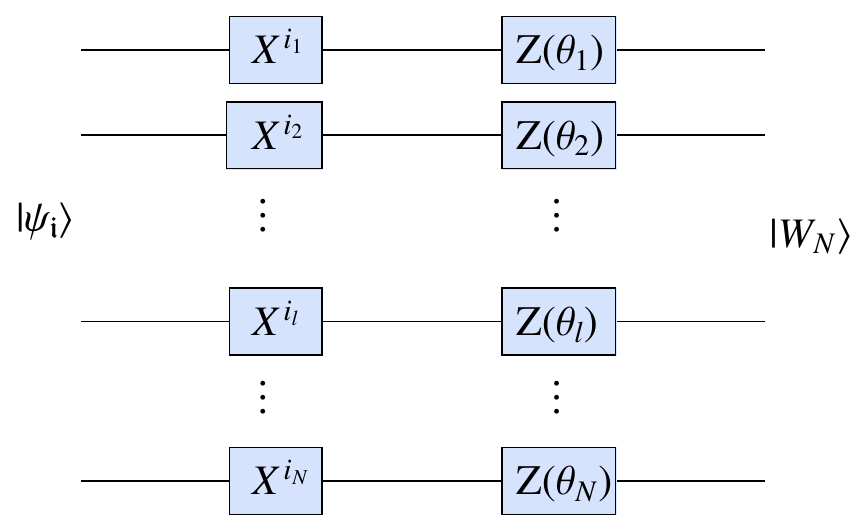}
	\caption{This is the circuit which transfer the state $|\psi_{\boldsymbol{i}}\rangle$ to $|W_N\rangle.$ }\label{fig:Wbasiscurcuit}
\end{figure}

Using the Matlab, we can calculate the rank of theirs corresponding quantities  to  check whether these sets are  locally stable or not.

\begin{example}\label{pro:W_SN_triqubit_four_to_seven}
	The set $\cS_N^{(W)}:=\{U_N^{(W)}|\boldsymbol{i}\rangle \ \mid \  \boldsymbol{i}\in \mathbb{Z}_2^N\}$ is  of  strongest  nonlocality  for $3\leq N\leq 8$. 
\end{example}	

In fact, by randomly choosing a subset of $\cS_N^{(W)}$ with $2^N-2^{N-2}$ elements, we find that it is  locally stable for each bipartition for $3\leq N\leq 8$.
We conjecture that the above statement should indeed hold for all $N$-qubit systems provided $N\geq 3$. Moreover, we have obtained some numerical results on the smallest set that can show the local stableness. Based on  our numerical results (for systems with small dimension), we conjecture that the bound in Theorem \ref{thm:AlgorithmBasis_Strong} is even compact! 

\begin{conjecture}\label{con:SN_size}
	Let $\mathcal{H}=\otimes_{i=1}^N\mathcal{H}_{A_i}$ be a $N$-parties quantum systems whose local dimension $\text{dim}_\mathbb{C}(\mathcal{H}_{A_i})=d_{i}\geq 2$. Then the following two statements hold
	\begin{enumerate}[(a)]
		\item 
		There exists some orthogonal set $\cS$ of pure states in $\mathcal{H}$  such that it is  locally  stable and  $|\cS|=\max_{i}\{d_i+1\}.$  That is,
		$$s(d_1,d_2,\cdots,d_N)= \max_{i}\{d_i+1\}.$$
		
		\item  There exists some orthogonal set $\cS$ of pure states in $\mathcal{H}$  such that it is of strongest   nonlocality and  	$|\cS| = \max_i\{\hat {d}_i+1\}.$ That is,
		$$S(d_1,d_2,\cdots,d_N)=  \max_i\{\hat {d}_i+1\}$$
		where $\hat{d}_i=(\prod_{j=1}^N d_j)/d_i$.
	\end{enumerate} 
\end{conjecture}

\vskip 12pt	

\section{Conclusion and Discussion}\label{sec:Conclusion}  In this paper, we studied a special class of sets with quantum nonlocality, i.e., the locally stable sets. That is, an orthogonal set of pure states in multipartite quantum system whose possible orthogonality preserving local  measurements are just trivial measurements. Locally stable sets are always locally indistinguishable sets. And we found that the two concepts are coincide  only  in two qubits systems.  We obtained an algebraic characterization of locally stable set. As a consequence, we obtained a lower bound of the cardinality on the locally stable set (and strongest   nonlocal set).  \add{ Moreover,  we  showed that locally stable sets of product states cannot be    perfect discrimination  under asymptotic LOCC  wherein an error is allowed but must vanish in the limit of an infinite number of rounds.}
			
			Moreover, we presented two constructions of sets that are of strongest   nonlocality.  Their proofs can be directly verified via   two basic  lemmas developed in Ref.  \cite{Shi21}.   One of the set contains genuinely entangled states except one fully product state. The other set  contains only genuinely entangled states. Our result give a complete answer to an open question raised in   Ref. \cite{Shi22N}.  This result gives an upper bound on the the smallest cardinality of those orthogonal sets in multipartite systems that are of strongest  nonlocality.   
			
			There are also some questions left to be considered. We conjectured that there is some set  of cardinality $\max_{i}\{d_i+1\} $ of  orthogonal states in $\otimes_{i=1}^N\mathcal{H}_{A_i}$ (where $d_i=\mathrm{dim}_{\mathbb{C}}(\mathcal{H}_{A_i})$) that is locally stable.  We also conjecture that there is some set  of cardinality $\max_{i}\{\hat{d}_i+1\} $ of  orthogonal states in $\otimes_{i=1}^N\mathcal{H}_{A_i}$ (where $d_i=\text{dim}_{\mathbb{C}}(\mathcal{H}_{A_i})$ and $\hat{d}_i=(\prod_{j=1}^N d_j)/d_i$) that is of strongest  nonlocality. In addition, it is also to consider the smallest sets of product states that are locally stable. We hope that the study of locally stable sets will enrich our understanding of the quantum nonlocality.

\vskip 8pt

\noindent\emph{Note added.}\, \,  Very rencently, the authors in Ref. \cite{Cao23} provided partial solutions to the part (a) of conjecture \ref{con:SN_size}.

\vspace{2.5ex}

\noindent{\bf Acknowledgments}\, \,   This  work  is supported  by  National  Natural  Science  Foundation  of  China  (12371458, 11901084, 12005092),  the Guangdong Basic
and Applied Basic Research Foundation under Grants   No. 2023A1515012074 and the   Science and Technology Planning Project of Guangzhou  under Grants   No. 2023A04J1296.

\end{document}